\documentclass{llncs}

 \usepackage{cite}

\usepackage{paralist}

\usepackage{csquotes}
\usepackage{microtype}
\usepackage{amsmath,amssymb,stmaryrd,upgreek,pifont}
\usepackage{graphicx}

\title{Contextuality in multipartite pseudo-telepathy graph games}
\titlerunning{Contextuality in multipartite pseudo-telepathy graph games} 

\author{Anurag Anshu\inst{1} \and 
Peter H{\o}yer\inst{2} \and 
Mehdi Mhalla\inst{3} \and  
Simon Perdrix\inst{4}}
\institute{ Centre for Quantum Technologies, National University of Singapore 
\and   University of Calgary, Canada 
\and Univ. Grenoble Alpes, CNRS, Grenoble INP, LIG, F-38000 Grenoble France 
\and CNRS, LORIA, Universit\'e de Lorraine, Nancy, France 
}

\newcommand{\ket}[1]{\left| #1 \right\rangle}
\newcommand{\bra}[1]{\left\langle #1 \right|}

\newcommand{\even}{\textup{Even}}
\newcommand{\odd}{\textup{Odd}}
\newcommand{\supp}{\textup{supp}}

\newcommand{\loc}{\textup{loc}}
\begin{document}

\maketitle
\begin{abstract}
Analyzing pseudo-telepathy graph games, we propose a way to build contextuality scenarios exhibiting the quantum supremacy using graph states. We  consider the combinatorial structures that generate equivalent scenarios. We introduce a new tool  called multipartiteness width to investigate which scenarios are harder to decompose and show  that there exist graphs generating scenarios with a linear multipartiteness width.
 \end{abstract}

\section{Introduction}

Contextuality is an active area of research that 
describes models of correlations and interpretations, and links to some fundamental questions about the natural world. It also provides a framework where one can utilize the understanding of quantum mechanics (and quantum information) in order to better analyze, understand, and interpret  macroscopic phenomena \cite{COLT,Context,livre,Bob,cwc}. 

The theoretical and experimental study of quantum world has proven that a scenario involving many parties (each having access to a local information) can contain correlations that do not possess any classical interpretation that relies on decomposition of these correlations using local functions. Contextuality can be viewed as a tool to describe the combinatorial structures present in these correlations.  

  A model of contextuality scenario is due to Abramsky and Brandenburger \cite{AB} and uses sheaf theory to naturally translate the consistency of interpretation by the pre-sheaf structure obtained by a distribution functor on the sheaf of events. 
  The authors introduce three levels of contextuality: ($i$) Probabilistic contextuality, which corresponds to the possibility of simulating locally and classically a probability distribution. It extends the celebrated Bell's theorem \cite{Bell} which shows that quantum probabilities are inconsistent with the predictions of any local realistic theory; ($ii$) Logical contextuality or possibilistic  contextuality, which extends Hardy's construction \cite{Hardy} and considers only the support of a probability distribution; ($iii$) Strong contextuality, which extends the properties of the GHZ state \cite {GHZ} and relies on the existence of a global assignment consistent with the support.  

  More recently Ac\'in, Fritz, Leverrier, and Bel\'en Sainz \cite{AFLS} have presented contextuality  scenarios defined as hypergraphs, in which vertices are called outcomes  and    hyperedges  are called measurements.  A general interpretation model is an assignment of non negative reals to the vertices that can be interpreted as a  probability distribution for any hyperedge  (weights of the vertices of each hyperedge sum to  1). Each hypergraph $H$ admits a set ${\cal C}(H)$ (resp. ${\cal Q} (H)$, ${\cal G} (H)$) of  classical (resp. quantum, general probabilistic) models  with ${\cal C}(H)\subseteq {\cal Q}(H)\subseteq{\cal G}(H)$.
  
  They have shown that the Foulis Randall product of hypergraphs \cite{FR}  allows one to describe the set of no-signaling models in product scenarios ${\cal G} (H_1\otimes H_2)$. They have also
  investigated  the multipartite case, 
  showing that the different products for composition produce models that are  observationally equivalent. 

   A particular case of contextuality scenarios is the \textit{pseudo-telepathy games} \cite{BB}, which are games that can be won by non-communicating players that share  quantum resources, but cannot be won classically without communication. 
A family of pseudo-telepathy games based on graph states have been introduced in \cite{AM}. 
The pseudo-telepathy game associated with a graph $G$ of order $n$ (on $n$ vertices), is a collaborative $n$-player game where each player receives a binary input (question) and is asked to provide, without communication, a binary output (answer). Some global pairs of (answers|questions) are forbidden  and correspond to losing positions.  
Given such a scenario, to quantify its multipartiteness,  we define the  {\bf multipartiteness width}:  a model on $n$ parties  has a multipartiteness width less than $k$ if it has an interpretation (assignment of real positive numbers to the vertices) that can  be obtained  using as ressources interpretations of  contextual scenarios  on less than $k$ parties. 



 It has been shown in \cite{BM} that  even though  GHZ type scenarios are maximally non local (strongly contextual), they can be won with 2 partite nonlocal boxes. So the multipartiteness width is different from the usual measures of contextuality \cite{qc,shane}. However, it has potential application for producing device independent witnesses for entanglement depth  \cite{LRB14}.

 


In section 2, we define the graph pseudo-telepathy games, investigate in detail the quantum strategy and link them to contextuality scenarios. We show in section 3  that provided that the players share multipartite randomness, it is enough to surely win the associated pseudo-telepathy game, in order to simulate the associated quantum probability distribution. In section 4,
we prove that graphs obtained by a combinatorial graph transformation called pivoting correspond to equivalent games. Finally,   we prove that there exist graphs for which the multipartiteness width is linear in the  number of players, improving upon the previous logarithmic bound given in \cite{AM}.

Note that even though the rules of these graph games appear non-trivial, they naturally correspond to the correlations present in outcomes of a quantum process that performs $X$ and $Z$ measurements on a graph state. Thus, they might be easy to produce empirically. Furthermore even if the space of events is quite large, the scenarios have the advantage of possessing concise descriptions, quite similar to the separating scenarios using Johnson graphs in \cite{AL}. Requiring such large structures to achieve possibilistic contextuality for quantum scenarios seems to be unavoidable. Indeed, it has been shown that multiparty XOR type inequalities involving two-body correlation functions cannot achieve pseudo-telepathy \cite{GRRH}.

\section{Pseudo-telepathy graph games, multipartiteness and contextuality scenarios}


\noindent{\bf Graph notations.} We consider finite simple undirected graphs. Let $G=(V,E)$ be a graph. For any vertex $u\in V$, $N_G(u) = \{v\in V~|~(u,v)\in E\}$ is the neighborhood of $u$. For any $D\subseteq V$, the odd neighborhood of $D$ is the set of all vertices which are oddly connected to $D$ in $G$: $\odd(D)=\{ v\in V :   |D\cap N(v) |=1 \bmod 2\}$. $\even(D)= V\setminus \odd(D)$ is the even neighborhood  of $D$, and $\loc(D)= D\cup \odd(D)$ is the local set of $D$ which consists of the vertices in $D$ and those oddly connected to $D$.  
For any $D\subseteq V$,  $G[D] = (D, E\cap D{\times} D)$ is the subgraph induced by $D$, and $|G[D]|$ its size, i.e. the number of edges of $G[D]$. 
Note that $\odd$ can be realized as linear map (where we consider subsets as binary vectors), which implies that for any two subset of vertices $A,B$, $\odd(A\oplus B)=\odd(A)\oplus\odd(B)$ where $\oplus$ denotes the symmetric difference.

We introduce the notion of \emph{involvement}: 
\begin{definition}[Involvement]
Given a graph $G=(V,E)$, a  set $D\subseteq V$ of vertices is \emph{involved} in a binary labelling $x\in \{0,1\}^V$ of the vertices if $ D\subseteq \supp(x) \subseteq \even(D)$, where $\supp(x)=\{u\in V, x_u=1\}$.
\end{definition}
 In other words,  $D$ is involved in the binary labelling $x$, if all the vertices in $D$ are labelled with $1$ and all the vertices in $\odd(D)$ are labelled with  $0$. 
Notice that when $G[D]$ is not a union of Eulerian graphs\footnote{ The following three properties are equivalent: (i) $D\subseteq \even(D)$; (ii) every vertex of $G[D]$ has an even degree; (iii) $G[D]$ is a union of Eulerian graphs. Notice that $D\subseteq \even(D)$ does not imply that $G[D]$ is Eulerian as it may not be connected.}, there is no binary labelling in which $D$ is involved. On the other hand, if $G[D]$ is a union of Eulerian graphs, there are $2^{|\even(D)|-|D|}$ binary labellings in which $D$ is involved. 

~

\noindent {\bf Collaborative games.} A multipartite collaborative game $\cal G$ for a set $V$ of players is a scenario characterised by a set $\mathcal L\subseteq \{0,1\}^V\times \{0,1\}^V$ of losing pairs: each player $u$ is asked a binary question $x_u$ and has to produce a binary answer $a_u$. The collaborative game is won by the players if for a given question $x\in \{0,1\}^V$ they produce an answer $a\in \{0,1\}^V$ such that the pair formed by $a$ and $x$, denoted $(a|x)$, is not a losing pair, i.e. $(a|x)\notin \cal L$.

A game is  pseudo-telepathic if classical players using classical resources cannot perfectly win the game (unless they cheat by exchanging messages after receiving the questions) whereas using entangled states as quantum resources the players can perfectly win the game, giving the impression to a quantum non believer that they are telepathic (as the only classical explanation to a perfect winning strategy is that they are communicating). 

{\bf {Example 1:}}\label{exa:mermin}
The losing set associated with the Mermin parity game \cite{Mermin} is  $ {\cal L}_{\textbf{Mermin}}=\{ (a|x) : \sum x_i=0 \bmod 2$ and $\sum a_i +  (\sum x_i)/2=1 \bmod 2 \}$. Notice that the losing set admits the following simpler description: ${\cal L}_{\textbf{Mermin}} = \{(a|x) : 2|a| = |x| + 2 \bmod 4\}$,  where $|x|=|\supp(x)|$ is the Hamming weight of $x$.


\noindent {\bf Collaborative graph games MCG(G)}: A multipartite collaborative game \textbf{MCG}$ (G)$ associated with a graph $G=(V,E)$, where $V$ is a set of players, is the collaborative game where the set of losing pairs is $ {\cal L}_G:=\{ (a|x) : \exists D$ involved in $x$ s.t. $\sum_{u\in \loc(D)} a_u = |G[D]|+1 \bmod 2 \}$. In other words, 
the collaborative game is won by the players if for a given question $x\in \{0,1\}^V$ they produce an answer $a\in \{0,1\}^V$ such that  for any non-empty $D$ involved in $x$, $\sum_{u\in \loc(D)} a_u = |G[D]|\bmod 2$. 


{\bf{Example 2:}} \label{exa:Kn}Consider \textbf{MCG}$ (K_n)$ the collaborative game associated with  the complete graph $K_n$ of order $n$. When a question $x$ contains an even number of $1$s the players trivially win since there is no non-empty subset of vertices involved in such a question. When $x$ has an odd number of $1$s, the set of players (vertices) involved in this question is $D =\supp(x)$. In this case, all the players are  either in $D$ or $\odd(D)$ thus the sum of all the answers has to be equal to $|G[D]|=\frac{|D|(|D|-1)}{2} = \frac{|D|-1}{2} \bmod 2$. 
Thus for the complete graph $K_n$,  $ {\cal L}_{K_n}=\{ (a|x) : |x|=1\bmod 2$ and $|a|= \frac{|x|-1}{2}+1\bmod 2\}=  \{ (a|x) :  2|a|= |x| +1 \bmod 4 \}$. 
Note that for this particular graph, the constraints are global in the sense that the sum of the answers of all the players is used for all the questions.
Notice also that the  set of losing pairs ${\cal L}_{K_n}=  \{ (a|x) :  2|a|= |x| +1 \bmod 4 \}$ is similar to the one of the Mermin parity game, ${\cal L}_{\textbf{Mermin}} = \{(a|x) : 2|a| = |x| + 2 \bmod 4\}$.  
In section \ref{secsimul}, we actually show the two games simulate each other.  


\noindent{\bf Quantum strategy (Qstrat):} In the following we show that for any graph $G$, the corresponding multipartite collaborative game can be won by the players if they share a particular quantum state. More precisely the state they share is the so-called graph state $\ket G= \frac{1}{\sqrt{2^{|V|}}}\sum_{y\in \{0,1\}^V} (-1)^{|G[\supp(y)]|}\ket y$, and they apply the following strategy:  every player $u$ measures his qubit according to $X$ if $x_u=1$ or according to $Z$ if $x_u=0$. 
Every player answers the outcome $a_u\in \{0,1\}$ of this measurement. 

This quantum strategy \textbf{QStrat}, not only produces correct answers, but provides all the good answers uniformly:

\begin{lemma}Given a graph $G=(V,E)$ and question $x \in \{0,1\}^V$, the probability $p(a|x)$ to observe the outcome $a\in \{0,1\}^V$ when each qubit $u$ of a graph state $\ket G$ is measured according to $Z$ if $x_u=0$ or according to $X$ if $x_u=1$ satisfies: 
$$p(a|x) = \begin{cases} 0 & \text{if $(a|x) \in \mathcal L$}\\
\frac{|\{\text{$D$ \textup{involved in} $x$}\}|} {2^{|V|}}&\text{otherwise.}
\end{cases}$$
\end{lemma}

\begin{proof}
According to the Born rule, the probability to get the answer $a\in \{0,1\}^V$ to a given question $x\in \{0,1\}^V$ is:
\begin{align*}
p(a|x)&=\bra G \left(\bigotimes_{v\in V\setminus \supp(x)} \frac {I+(-1)^{a_v}Z_v}2\right)\otimes \left(\bigotimes_{u\in \supp(x)} \frac {I+(-1)^{a_u}X_u}2\right)   \ket G\\
&= \frac 1{2^n} \sum_{D\subseteq V} (-1)^{\sum_{u\in D}a_u} \bra GZ_{D \setminus \supp(x)} X_{D\cap \supp(x)}\ket G
\end{align*}
The basic property which makes this strategy work is that  for any $u\in V$, $X_u \ket G = Z_{N(u)}\ket G$. 
As a consequence, since $X$ and $Z$ anti-commute and $X^2=Z^2=I$, for any $D\subseteq V$, $X_D\ket G = (-1)^{|G[D]|}Z_{\odd(D)}\ket G$. Thus, 
\begin{align*}
p(a|x)
&= \frac 1{2^n} \sum_{D\subseteq V} (-1)^{|G[D\cap \supp(x)]| + \sum_{u\in D}a_u} \bra GZ_{(\odd(D\cap \supp(x)))\oplus(D\cap \setminus \supp(x))} \ket G
\end{align*}
Where $\oplus$ denotes the symmetric difference. Since $\bra G Z_C\ket G =\begin{cases} 1& \text{if $C=\emptyset$}\\0& \text{otherwise}\end{cases}$, 
\begin{align*}
p(a|x)&= \frac 1{2^n} \sum_{D\subseteq V, D\setminus  \supp(x) = \odd(D\cap \supp(x))} (-1)^{|G[D\cap \supp(x)]| +\sum_{u\in D}a_u}\\
&=\frac 1{2^n} \sum_{D_1\subseteq \supp(x)}\sum_{D_0 \subseteq V\setminus \supp(x), D_0= \odd(D_1)} (-1)^{|G[D_1]| +\sum_{u\in D_0\cup D_1}a_u}\\
&= \frac 1{2^n} \sum_{D_1\subseteq \supp(x), \odd(D_1)\cap \supp(x)=\emptyset} (-1)^{|G[D_1]| +\sum_{u\in \loc(D_1)}a_u}\\
&= \frac 1{2^n} \sum_{D_1 \text{involved in } x} (-1)^{|G[D_1]| +\sum_{u\in \loc(D_1)}a_u}= \frac {|R^{(x,a)}_0| - |R^{(x,a)}_1|}{2^n} 
\end{align*}
where $R^{(x,a)}_d = \{\text{$D$ involved in $x$}:|G[D]| +\sum_{u\in \loc(D)}a_u = d \bmod 2\}$.  If $(a|x)\notin \mathcal L$, then $R^{(x,a)}_1=\emptyset$, so $p(a|x)=\frac{|\{\textup{$D$ involved in $x$}\}|} {2^n}>0$ since $\emptyset$ is involved in $x$. Otherwise, there exists $D'\in R^{(x,a)}_1$. Notice that $R^{(x,a)}_0$ is a vector space ($\forall D_1, D_2\in R^{(x,a)}_0, D_1\oplus D_2 \in R^{(x,a)}_0$) and $R^{(x,a)}_1$ an affine space $R^{(x,a)}_1 = \{D'\oplus D~|~D\in R^{(x,a)}_0\}$. Thus $|R^{(x,a)}_0| = |R^{(x,a)}_1|$ 
which implies $p(a|x)=0$. 
\qed \end{proof}

The probability distribution produced by \textbf{QStrat}  depends on the number of sets $D$ involved in a given question $x$. Notice that a set $D\subseteq \supp(x)$ is  involved in $x$ if and only if $D\in \textup{Ker}(L_x)$, where $L_x$ linearly\footnote{$L_x$ is linear for the symmetric difference: $L_x(D_1\oplus D_2)=L_x(D_1)\oplus L_x(D_2)$} maps $A\subseteq \supp(x)$ to $\odd(A)\cap \supp(x)$. Thus $|\{\text{$D$ \textup{involved in} $x$}\}|=2^{|x| - rk_G(x)}$, where $rk_G(x) = \log_2(|\{L_x(A) : A\subseteq \supp(x))\}|)$ is the rank of $L_x= A\mapsto \odd(A)\cap \supp(x)$.

~

\noindent {\bf Contextuality scenario.}  Following  the hypergraph model of  \cite{AFLS}, we associate with every graph $G$ a contextuality scenario, where each vertex is a pair $(a|x)$ and each hyperedge corresponds, roughly speaking,  to a constraint. There are  two kinds of hyperedges, those ($H_{\textup{Nsig}_V}$) which guarantee no-signaling and those ($H_G$), depending on the graph $G$, which avoid the losing pairs.  

\begin{itemize}
\item  $H_{\textup{Nsig}_V}$ is the hypergraph representing the no-signaling polytope. It corresponds \cite{AFLS} to the Bell scenario   $B_{V,2,2}$ where $|V|$ parties have access to $2$ local measurements each, each of which has $2$ possible outcomes (see Figure \ref{fig:Hnsig2}), which is obtained as a product\footnote{The Foulis Randall product of scenarios \cite{AFLS} is the scenario $H_A \otimes H_B$ with vertices $V(H_A\otimes H_B)=V(H_A)\times V(H_B)$ and edges $E(H_A\otimes H_B)=E_{A\rightarrow B}\cup E_{A\leftarrow B}$ where  $E_{A\rightarrow B}:=\{\cup_{a\in e_A} \{a\}\times f(a) : e_a\in E_A, f: e_A \rightarrow E_B\}$ and $E_{A\leftarrow B}:=\{\cup_{b\in e_A} f(b)\times \{b\} : e_b\in E_b, f: E_B \rightarrow E_A\}$. In the multipartite case there are several ways to define products, however they all correspond to the same non-locality constraints \cite{AFLS}. Therefore one can just consider the minimal product $ ^{\min}\otimes_{i=1 }^n H_i$ which has vertices in the cartesian product $V=\Pi V_i$ and edges $\cup_{k\in[1,n]} E_k$ where $E_k=\{(v_1\ldots ,v_n), v_i\in e_i \, \forall i\neq k,\, v_k\in f(\overrightarrow{v})\}$ for some edge $e_i\in E(H_i)$ for every party $i\neq k$ and a function $\overrightarrow{v}\mapsto   f(\overrightarrow{v})$ which assigns to every joint outcome $\overrightarrow{v}=(v_1\ldots v_{k-1},v_{k+1},\ldots v_n)$ an edge $f(\overrightarrow{v})\in E(H_k)$ (the $k^{th}$  vertex is replaced by a function of the others).} of the elementary scenario $B_{1,2,2}$.


\begin{figure}
\centering
\begin{minipage}{.5\textwidth}
  \centering
  \includegraphics[width=.6\linewidth]{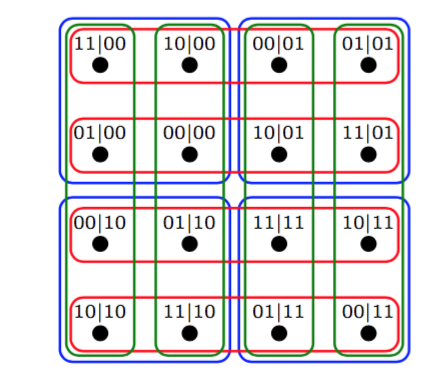}
  \caption{$H_{\textup{Nsig}_2}$: hyperedges of the Bell scenario $B_{2,2,2}$ from \cite{AL}}
  \label{fig:Hnsig2}
\end{minipage}%
\begin{minipage}{.5\textwidth}
  \centering
  \includegraphics[width=.6\linewidth]{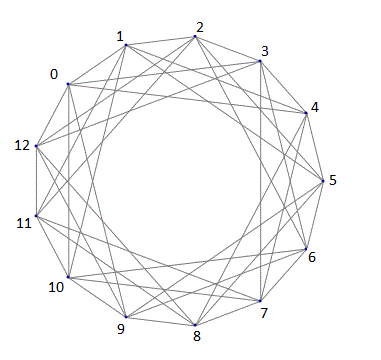}
  \caption{Paley Graph of order 13}
  \label{fig:paley}
\end{minipage}
\end{figure}

\item 
The hypergraph $H_G$ defined on the same vertex set, corresponds to the game constraints:    for each 
question\footnote{Note that for the questions $x$ for which there exists no $D$ involved in $x$, all the answers are allowed thus the constraints represented by the associated edge is a hyperedge of no-signaling scenario $H_{Nsig}$.} $x\in \{0,1\}^V$ we associate an hyperedge $e_x$ containing all the answers which make the players win on  $x$ i.e., $e_x=\{(a|x)\in \{0,1\}^V\times \{0,1\}^V, (a|x)\notin \cal L\}$.

\end{itemize}


Given  a graph $G=(V,E)$, \textbf{MCG}$ (G)$ is a \emph{pseudo-telepathy} game if it admits a quantum model (${\cal Q} (H_G\cup H_{\textup{Nsig}_V})\neq \emptyset$) but no classical model (${\cal C}(H_G\cup H_{\textup{Nsig}_V})=\emptyset$). It has been proven in  \cite{AM} that \textbf{MCG}$(G)$ is pseudo-telepathic if and only if $G$ is not bipartite.

{\bf { Example 3:}}
In a complete graph $K_n$ of order $n$, there exists a non-empty set $D$ involved in a question $x\in \{0,1\}^V$ if and only if $|x|=1\bmod 2$.  With each such question $x$, the associated hyperedge is $e_x=\{(a|x)\in \{0,1\}^V\times \{0,1\}^V s.t.~2|a|\neq|x|+1\bmod 4\}$. 

\sloppy
 {\bf { Example 4:}}
 In the graph  Paley 13 (see Figure \ref{fig:paley}),  $\odd(\{0,1,4\})=\{2,7,8,9,11,12\}$ thus  if $\{0,1,4\}$ is involved in $x$  {\em i.e.}~$x_i=1$ for $i\in \{0,1,4\}$ and $x_i=0$ for $i\in \{2,7,8,9,11,12\}$ then the associated pseudo-telepathy game requires that the sum of the outputs of these nine players $\sum_{i\not \in \{3,5,6,10\}} a_i$ has to be odd. This corresponds to 8 hyperedges $e_{jkl}$  for $j,k,l \in \{0,1\}$ in the contextuality scenario  where $e_{jkl}=\{(a|x),\, \sum_{i\not \in \{5,6,10\}} a_i=1 \bmod 2,\, x_i=1$ for $ i\in \{0,1,4\}$, $x_i=0$ for $i\in \{2,7,8,9,11,12\},\, x_5=j, \, x_6=k, \, x_{10}=l \}$.



The probabilistic contextuality is what was considered in \cite{AM}  as it corresponds to investigating the possibility of simulating a probability distribution of a quantum strategy playing with graph states.
The two other levels of contextuality gain some new perspectives when viewed as games: indeed the possibilistic contextuality coincides with 
the fact that the players cannot give all the good answers with non zero probability using classical local strategies, and  strong contextuality  just  means that classical players cannot win the game (even by giving a strict subset of the good answers).

\begin{definition}
An interpretation $p:\{0,1\}^V\times \{0,1\}^V \to [0,1]$ is $k$-multipartite if it can be obtained  by a strategy without communication using  nonlocal boxes that are at most $k$-partite: for any set $I\subset V$  with $|I|\le k$, each player has access to one bit of a variable $\lambda_I(a_I|x_I)$ that has a no-signaling probability distribution.

\end {definition}


In other words, a $k$-multipartite interpretation can be obtained  
with no-signaling correlations involving
 at most $k$ players. For example the strategy to win the Mermin game proposed in \cite{BM} where each pair among $n$ players share a (2-partite) non localbox and  each player outputs the sum of his boxes' ouputs is a 2-multipartite interpretation. Similarly, the result in \cite{BP} where they prove that a probability distribution that can be obtained by 5 players measuring a quantum state cannot be simulated without communication using any number of bi-partite non local boxes shows that it is not a  2-multipartite interpretation.\footnote{The probability distribution described in \cite{BP} corresponds to the quantum winning strategy on the graph state obtained from a cycle with 5 vertices.} 

\begin{definition}[multipartiteness width]
A scenario has a multipartiteness width $k$ if it admits a $k$-multipartite interpretation but no $(k-1)$-multipartite interpretation. 
\end{definition}


In a contextual scenario, the more hyperedges one adds the less possible interpretations exist.  A scenario has a multipartiteness width $k$  if its hyperedges already forbids  all the interpretations of a product of  Bell scenarios on less than $k$ parties.
 For a scenario, having a classical interpretation  means being decomposable : one can think of the probability distribution as local actors acting each on his bit  and that's a classical interpretation. The multipartiteness width measure how non-decomposable a scenario is : it can not be decomposed with interpretations where each subspace has a small width.

It implies that the players cannot perfectly win the game if they  have only quantum systems on less than $k$ qubits, this corresponds to using $k$ seprable states as ressources as defined in \cite{GTB05}.






Note that  from the observations in  \cite{AM} the multipartiteness width of the scenario generated by the Paley graph on 13 (see figure \ref{fig:paley}) is strictly larger than 4.




 In the next section, we will show how for the scenarios we describe,  being able to give only good answers allows for simulation of the quantum distribution with random variables.
Thus, the contextuality lies in the combinatorial structure of the graph and the three levels collapse for these games.

\section{Simulating a probability distribution is the same as winning the pseudo-telepathy graph game}
In \cite{AM} it was proven that for some graphs, the probability distributions of the quantum strategy using the graph states cannot be simulated using non local boxes on less than $k$ parties, we show here that any strategy that allows to win the game can be extended using random variables shared between neighbors (in the graph) to simulate this uniform probability distribution coming from the quantum strategy.
%
%
%

We start by describing a classical strategy  \textbf{CStrat}  based on shared random variables rather than quantum states. We show that 
\textbf{CStrat}   is a winning strategy if and only if the graph is bi-partite. We also show that \textbf{CStrat} 
 can be used to make any winning strategy a uniform winning strategy, i.e. each valid answer to a given question are equiprobable. 
We show that \textbf{CStrat}   can be locally adapted to collaborative games on graphs that can be obtained by a sequence of local complementations.

\noindent {\bf Classical strategy (Cstrat): } 
\label{def:}
Given a graph $G=(V,E)$, pick uniformly at random $\lambda\in \{0,1\}^V$. 
Each player $u\in V$ receives a pair of bits $(\lambda_u,\mu_u)$, where $\mu_u=\sum_{v\in N_G(u)} \lambda_u\bmod 2$.  Given a question $x\in \{0,1\}^V$, each player $u\in V$ locally computes and answers $a_u =(1-x_u).\lambda_u + x_u.\mu_u\bmod 2$. 
\sloppy

\begin{lemma}
Given a graph $G=(V,E)$ and a question $x\in \{0,1\}^V$, 
\textbf{CStrat}  produces an answer uniformly at random in $\{a\in \{0,1\}^V~|~ \exists D\subseteq S, (A\oplus \odd(A \oplus D)){\cap} S = \emptyset \text{ where $A = \supp(a)$ and $S=\supp(x)$}\}$. 
\end{lemma}

\begin{proof} Given a graph $G=(V,E)$, a question $x\in \{0,1\}^V$ and  $a\in \{0,1\}^V$, the probability that \textbf{CStrat}  outputs $a$ is 
\begin{eqnarray*}
p(a|x) &=&p\left(\forall u\in V{\setminus} S,  a_u = \lambda_u\right) p(\forall u\in S, a_u = {\sum_{v\in N(u)}\lambda_v} \bmod 2~|~\forall u \in V{\setminus} S,  a_u = \lambda_u)\\
&=&p\left(A\setminus S = \Lambda \setminus S\right) p(A\cap S= \odd(\Lambda) \cap S |   A\setminus S = \Lambda \setminus S)
\end{eqnarray*}
where $S=\supp(x)$, $A = \supp(a)$ and $\Lambda = \supp(\lambda)$.
Since $p\left(A\setminus S = \Lambda\setminus S\right) = \frac{1}{2^{n-|x|}}$, 
\begin{eqnarray*}
p(a|x) 
&=&\frac{1}{2^{n-|x|}}p\left(A\cap S= \odd(\Lambda\cap S \oplus \Lambda \setminus S) \cap S |   A\setminus S = \Lambda\setminus S\right)\\
&=&\frac{1}{2^{n-|x|}}p\left(A\cap S= \odd(D \oplus (A \setminus S)) \cap S |   A\setminus S = \Lambda\setminus S\right)
\end{eqnarray*}
where $D = \Lambda\cap S$. 
If $ A\cap S\neq  \odd(D \oplus (A \setminus S)) \cap S$  for all $D\subseteq S$, then $p(a|x)=0$. Otherwise, the set of subsets $D$ of $S$ which satisfy the condition is the affine space $\{D_0\oplus D |\, D\subseteq S \wedge  \odd(D)\cap S = \emptyset \}$, where $D_0$ is a fixed set which satisfies $A\cap S= \odd(D_0 \oplus A \setminus S) \cap S$. Thus the $p(a|x)=\frac{1}{2^{n-|x|}}.\frac{|\{D\subseteq S | \odd(D)\cap S = \emptyset\}|}{2^{|x|}} = 2^{|x|-rk_G(x)-n}$, which is independent of $a$, proving the uniformity of the answer. 

Finally notice   $\exists D_0\subseteq S, A\cap S= \odd(D_0 \oplus (A \setminus S)) \cap S$ if and only if 
$\exists D_1\subseteq S, (A\oplus \odd(A \oplus D_1)){\cap} S = \emptyset$, by taking $D_1 = D_0\oplus (A\cap S)$. 
\qed \end{proof}

We consider some standard graph transformations : 
Given a graph $G=(V,E)$ the local complementation on a vertex $u\in V$ produces the graph $G*u=(V,E\oplus K_{N(v)})$ where the sum is taken modulo 2  (it is the symmetric difference) and $K_U$ is the complete graph on $U\subset V$. $G*u$ is obtained from $G$ by  exchanging the edges by non edges and vice versa in the neighborhood of the vertex $u$. Pivoting using an edge $(u,v)$, is a sequence of three local complementations  $G\wedge uv=G*u*v*u$. 
We denote by $\delta_{loc}(G)$ ($\delta_{piv}(G)$) the minimum degree taken over all graphs that can be obtained from $G$
through some sequence of local complementations (edge pivots).

Given the shared randomness $(\lambda_v,\mu_v)_{v\in V}$ associated with $G$, if player $u$ replaces its first bit by the XOR of its two bits, and each of his neighbors replaces his second bit by the XOR of his two bits, one gets the shared randomness associated with $G*u$. 
\begin{lemma}
\label{lem:delta}
Given the probability distribution $(\lambda_v,\mu_v)_{v\in V}$ associated with $G$, if player $u$ replaces its first bit by the XOR of its two bits, and each of its neighbors replaces their second bit by the XOR of their two bits, one gets the probability distribution associated with $G*u$.
\end{lemma}


\begin{proof}
\sloppy
For any $v\in V$, let $\lambda'_v=\begin{cases}\lambda_u+\mu_u\bmod 2&\text{if $u=v$}\\\lambda_v&\text{otherwise}\end{cases}$ and

$\mu'_v=\begin{cases}\mu_v+\lambda_v\bmod 2&\text{if $v\in N_G(u)$}\\\mu_v&\text{otherwise}\end{cases}$. 

Since $N_{G*u}(v) =\begin{cases}N_G(v)\oplus N_G(u)\oplus \{v\}&\text{if $v\in N_G(u)$}\\N_G(v)&\text{otherwise}\end{cases}$,   one gets that for any $v\notin N_G(u)$, $\mu'_v=\mu_v=\sum_{w\in N_G(v)}\lambda_w \bmod 2 = \sum_{w\in N_{G*u}(v)}\lambda'_w\bmod 2$ and  for any $v\in N_G(u)$, $\mu'_v = \mu_v +\lambda _v=\lambda_v+\sum_{w\in N_G(v)}\lambda_w \bmod 2=\lambda_v+ \mu_u + \sum_{w\in N_G(v)}\lambda'_w \bmod 2=\lambda_v+\sum_{w\in N_G(u)}\lambda'_w+\sum_{w\in N_G(v)}\lambda'_w \bmod 2=\sum_{w\in N_{G*u}(v)}\lambda'_w \bmod 2$. 
\qed \end{proof}

Thus the probability distribution   corresponding to the classical strategy for $G$ can be locally transformed into the probability distribution associated with the $G*u$, thus one can use local complementation to optimise the cost of preparing the shared randomness. For instance the classical strategy \textbf{CStrat}  for a graph $G$ requires  shared random bits on at most $\Delta_{loc}(G)+1$ players, where $\Delta_{loc}(G)=\min(\Delta(G'), \text{ s.t. } \exists u_1,\ldots, u_k, G'=G*u_1*\ldots *u_k)$ and $\Delta(G)$ is its maximum degree. If there is no pre-shared random bits, the probability distribution can be prepared using at most $2|G|_{loc}$ communications in-between the players, where $|(G)|_{loc} = \min(|G'|, \text{ s.t. } \exists u_1,\ldots, u_k, G'=G*u_1*\ldots *u_k)$ is the minimum number of edges by local complementation.


Now we show how, using the classical strategy \textbf{CStrat} , one can simulate the quantum strategy \textbf{QStrat}  given an oracle that provides only good answers.

\begin{lemma}
\label{lem:simwin} 
For any collaborative game on a graph  $G$, for any  strategy  $Q$ that never loses, there exists a strategy $Q'$ using the outputs of $Q$ and shared random variables that simulate \textbf{QStrat} .
\end{lemma}

\begin{proof}

Given a collaborative graph game on a graph $G$,
let  $Q$ be a strategy that always outputs permissible outputs  for any set of inputs $x$, so we have pairs $(a|x)\not \in \mathcal L$.
We consider the strategy which combines $Q$ and \textbf{CStrat}  for this graph: For a given question $x$, $Q'$ outputs the XOR of the $Q$ answer and \textbf{CStrat}  answer for $x$. First we prove that such an answer is a valid answer and then the uniform probability among the possible answer to a given question.
Given a question $x\in \{0,1\}^V$, suppose $Q'$ outputs $a'\in \{0,1\}^V$:
$\forall u\in V$, $a_u' = a_u+(1-x_u)\lambda_u +x_u \mu_u$ where $a_u$ is the answer produced by $Q$ and $\lambda$ and $\mu$ are as defined in the classical strategy.


%
%
%
%
%
%
%

By contradiction, assume $(a'|x)\in \mathcal L$, so there exists $D$ involved in $x$ such that $\sum_{u\in \loc(D)}a'_u =|G[D]|+1 \bmod 2 $

\begin{eqnarray*}
\sum_{u\in \loc(D)} a'_u &=& \sum_{u\in \loc(D)}  \left(a_u + (1-x_u)\lambda_u+x_u\mu_u\right)\bmod 2 \\
&=&\sum_{u\in \loc(D)}a_u+\sum_{u\in \loc(D) \setminus \supp(x)} \lambda_u+\sum_{u\in \loc(D) \cap \supp(x)}\mu_u \bmod 2 \\
&=&\sum_{u\in \loc(D)}a_u+\sum_{u\in \odd(D)} \lambda_u+\sum_{u\in D}\sum_{v\in N(u)} \lambda_v  \bmod 2 \\
&=&\sum_{u\in \loc(D)}a_u+\sum_{u\in \odd(D)} \lambda_u+\sum_{v\in \odd(D)} \lambda_v  \bmod 2 \\
&=&\sum_{u\in \loc(D)}a_u \bmod 2
\end{eqnarray*}

Thus $(a|x)\in \mathcal L$ which is a contradiction thus $p(a'|x)=0$ if $(a'|x) \in \mathcal L$

Now we prove that $p(a'|x) = 2^{|x|-n-rk_G(x)}$.
First assume $Q$ is determinist, thus $p(a'|x)$ is the probability that the classical strategy outputs $a+ a':=(a_u+a'_u\bmod 2)_{u\in V}$. Since  this probability is non zero it must be $2^{|x|-n-rk_G(x)}$.  If $Q$ is probabilistic, $p(a'|x) = \sum_{a\in \{0,1\}^V} p(\text{$Q$ outputs $a$ on $x$})p(\text{classical strategy outputs $a+a'$ on $x$})\le  2^{|x|-n-rk_G(x)}\sum_{a\in \{0,1\}^V} p(\text{$Q$ outputs $a$ on $x$})\le 2^{|x|-n-rk_G(x)}$. 
Thus each answer $a$ produced by the strategy on a given question $x$ is s.t. $(a|x)\notin \mathcal L$ and occurs with probability at most $2^{|x|-n-rk_G(x)}$. Since $|\{a\in \{0,1\}^V~|~(a|x)\notin \mathcal L\}| = 2^{|x|-n-rk_G(x)}$, each of the possible answers is produced by the strategy and occurs with probability $2^{|x|-n-rk_G(x)}$. 
\qed \end{proof}
%
%
%
%

\section{Locally equivalent games}
\label{secsimul}

A pseudo telepathy game $\cal G$ locally simulates another pseudo telepathy game $\cal G'$ if any winning strategy for $\cal G$ can be locally turned  into a winning strategy for $\cal G'$:

\begin{definition}[Local Simulation]
Given two pseudo telepathy games $\cal G$ and $\cal G'$ on a set $V$ of players which sets of losing pairs are respectively $\cal L_G$ and $\cal L_{G'}$, 
$\cal G$ locally simulates $\cal G'$ if for all $u\in V$, 
there exist $f_1,\ldots, f_n: \{0,1\} \to \{0,1\}$ and $g_1,\ldots,g_n: \{0,1\} \times \{0,1\} \to \{0,1\}$ s.t. $\forall x,a \in \{0,1\}^V$
$(g(a,x),x)\in {\cal L_{G'}} \Rightarrow (a | f(x))\in \cal L_G$  
where $f(x) = (f_u(x_u))_{u\in V}$ and $g(a,x) = (g_u(a_u,x_u))_{u\in V}$.
%
%
\end{definition}
 
Assuming $\cal G$ locally simulates $\cal G'$ and that the players have a strategy to win $\cal G$, the strategy for $\cal G'$ is as follows: given an input $x$ of $\cal G'$, each player $u$ applies the preprocessing $f_u$ turning her input $x_u$ into $f_u(x_u)$, then they collectively play the game $\cal G$ with this input $f(x)$ getting an output $a$ s.t. $(a|f(x))\notin \cal L_{G}$. Finally each player $u$ applies a postprocessing $g_u$ which depends on her output $a_u$ and her initial input $x_u$ to produce the output $g_u(a_u,x_u)$ to the game $\cal G'$. This output is valid since, by contradiction,  $(g(a,x),x)\in {{\cal L}_{\cal G'}}$ would imply $(a|f(x))\in \cal L_{G}$.

%
%
\begin{definition}[Local Equivalence]
$\cal G$ and $\cal G'$ are locally equivalent games if $\cal G$ locally simulates $\cal G'$ and $\cal G'$ locally simulates $\cal G$. 
\end{definition}

In the following we give two examples of locally equivalent games : first we show that the games associated with the complete graphs are locally equivalent to Mermin parity games, and then that pivoting, a graph theoretical transformation, produces a graph game locally equivalent to the original one:

\begin{lemma}
For any $n$, the game associated with the complete graph $K_n$ is locally equivalent to the Mermin parity game on $n$ players.
\end{lemma}
\begin{proof}
The set of losing pairs of the two games are  ${\cal L}_{\textsf{Mermin}} =\{(a|x), 2|a| = |x|+2\bmod 4\}$  (see example 1 
) and ${\cal L}_{K_n} =\{(a|x), 2|a| = |x|+1\bmod 4\}$ (see example 2 
). \\{[$K_n$ simulates Mermin]} Let $u\in V$ be a fixed player. We define for all $v\in V$, $f_v(x_v) = \begin{cases}1-x_v & \text{if $v=u$}\\ x_v&\text{otherwise}\end{cases}$ and $g_v(a_v,x_v) = \begin{cases}a_v+x_v-2a_vx_v & \text{if $v=u$}\\ a_v&\text{otherwise}\end{cases}$. If $(g(a,x),x)\in {\cal L}_{\textsf{Mermin}}$ then $2|g(a,x)| = |x|+2\bmod 4$ which implies $2(|a|-a_{u}+a_u+x_{u}-2a_ux_u)= 2|a| + 2x_u = |x|+2\bmod 4$, so $2|a| = |x|+2-2x_u = |f(x)|+1\bmod 4$, thus $(a|f(x))\in {\cal L}_{K_n}$. \\
{[Mermin simulates $K_n$ ]} Let $u\in V$ be a fixed player. $f$ and $g$ are defined like in the previous case  except $g_u(a_u,x_u) = 1-a_u-x_u+2a_ux_u$.  If $(g(a,x),x)\in {\cal L}_{K_n}$ then $2|g(a,x)| = |x|+1\bmod 4$ which implies $2(|a|-a_u+1-a_u-x_u+2a_ux_u) = 2|a| +2 -2x_u= |x|+1\bmod 4$, so $2|a| = |x|+1-2x_u +2 = |f(x)|+2\bmod 4$, thus $(a|f(x))\in {\cal L}_{\textsf{Mermin}}$.   
\qed \end{proof}

\begin{lemma}Given a graph $G=(V,E)$ and $(u,v)\in E$, the games associated with $G$ and $G\wedge uv$ are locally equivalent. 
\end{lemma}

\begin{proof} Since pivoting is its self inverse ($(G\wedge uv)\wedge uv =G$) it is enough to prove that $G$ locally simulates $G\wedge uv$. The proof is based on the existence of a quantum strategy for any graph $G$ which consists in sharing the quantum state $\ket G$, and for each player $w$, in measuring her qubit according to $X$ if $x_w=1$ or according to $Z$ if $x_w=0$, and then output the outcome $a_w\in \{0,1\}$ of this measurement. A quantum strategy for $G$ can be turned into a quantum strategy for $G\wedge uv$ due to the following property of graph states: $\ket{G\wedge uv}= H_{u,v}Z_{N(u)\cap N(v)} \ket G$ \cite{MP}. The unitary map $H$ exchanges $X$- and $Z$- measurements -- i.e. for any state $\ket \phi$, apply $H$ on $\ket \phi$ followed by a $Z$-measurement (resp. $X$) produces the same classical outcome as measuring $\ket \phi$ according to $X$ (resp. $Z$) --  while the unitary $Z$ leaves invariant the classical outcomes of a $Z$-measurement and exchanges the two possible outcomes of a $X$-measurement. As a consequence, $(g(a,x),x)\in {{\cal L}_{G\wedge uv}}\Rightarrow (a|f(x))\in {\cal L}_{G}$, where $f$ and $g$ are defined as follows 
 $f_w(x_w)=\begin{cases}1{-}x_w&\text{if $w{\in} \{u,v\}$}\\ x_w&\text{otherwise}\end{cases}$ and $g_w (a_w,x_w) = \begin{cases}a_w{+} x_w{-}2a_wx_w&\text{if $w{\in} N(u){\cap}N(v)$}\\a_w&\text{otherwise}\end{cases}$
\qed \end{proof}


Therefore, the important quantity for the pre-shared randomness for  the strategies defined with a graph  is
$\Delta_{piv}(G)=\min \{\Delta(G'), G' \text{ pivot equivalent to $G$}\}$.

\section{Scenarios  with  linear multipartiteness width}

We prove that there exist contextuality scenarios with linear multipartiteness width. We use a graph property called $k$-odd domination which is related \cite{AM} to the classical simulation of the quantum probability distribution obtained by playing the associated graph game.
Since bipartite graphs correspond to graph games that can be won classically \cite{AM}, we focus on the non-bipartite case by showing that there exist non-bipartite 0.11$n$-odd dominated graphs of order $n$.

\begin{definition}[$k$-odd domination \cite{AM}]
A graph $G = (V,E)$ is $k$-odd dominated (k-o.d.) iff for any $S\in {V \choose k}$, there exists a labelling of the vertices in $S = \{v_1, \ldots, v_k\}$ and $C_1, \ldots C_k$, s.t. $\forall i$,  
$C_i\subseteq V\setminus S$ and 
$\odd(C_i)\cap \{v_i, \ldots v_k\} = \{v_i\}$ and 
$C_i\subseteq \even(C_i)$.
\end{definition}
\begin{lemma}\label{lem1}
For any $k\ge 0$, $r\ge 0$ and any graph $G=(V,E)$ a graph of order $n$ having two distinct independent sets $V_0$ and $V_1$ of order $|V_0|=|V_1|=\lfloor \frac{n-r}2\rfloor$, 
$G$ is $k$-odd dominated if for any $i\in \{0,1\}$, and any  non-empty $D\subseteq V\setminus V_i$,
 $|Odd_G(D)\cap  V_{i}|>{k-|D|}$ 
\end{lemma}
\begin{proof}
Given $S_0 \subseteq V_0$, $S_1 \subseteq V_1$, and $S_2\subseteq V_2=V\setminus (V_0\cup V_1)$ s.t. $|S_0| + |S_1|+|S_2|=k$, we show that for any $u\in S=S_0\cup S_1\cup S_2$, there exists $C_u\subseteq  V\setminus S$ s.t. $Odd(C_u)\cap S = \{u\}$ and $C_u\subseteq Even(C_u)$. 
For any $u\in S$, there exists $i\in \{0,1\}$ s.t.  $u\in S_i\cup S_2$. 
Let $L_i: 2^{S_i\cup S_2} \to 2^{V_{1-i}\setminus S_{1-i}}$ be the function which maps $D\subseteq S_i\cup S_2$ to $L_i(D) = Odd_G(D)\cap (V_{1-i}\setminus S_{1-i})$.

$L_i$ is linear according to the symmetric difference. $L_i$ is injective: for any $D\subseteq S_i\cup S_2$, $Odd(D) \cap (V_{1-i}\setminus S_{1-i})=\emptyset$ implies $Odd(D)\cap V_{1-i} \subseteq S_{1-i}$, thus $|Odd(D)\cap V_{1-i} |\le |S_{1-i}|$.
notice that $|D|\le |S_i|+|S_2|$, so $|Odd(D)\cap V_{1-i}|\le |S_{1-i}| \le |S_0|+|S_1|+|S_2|-|D| = k-|D|$, so 
$D=\emptyset$. 

The matrix representing $L_i$ is nothing but the submatrix $\Gamma_{[S_i\cup S_2, V_{1-i}\setminus S_{1-i}]}$ of the adjacency matrix $\Gamma$ of $G$. So its transpose 
$\Gamma_{[V_{1-i}\setminus S_{1-i},S_i\cup S_2]}$ is surjective which means that the corresponding linear map $L_i^T:2^{V_{1-i}\setminus S_{1-i}} \to 2^{S_i\cup S_2}= C\mapsto Odd_G(C)\cap (V_{1-i}\setminus S_{1-i})$ is surjective, so $\exists C_u\subseteq V_{1-i}\setminus S_{1-i}$ s.t. $Odd_G(C_u)\cap (S_i\cup S_2)=\{u\}$, which implies, since $V_{1-i}$ is an independent set, that $Odd_G(C_u)\cap S = \{u\}$ and $C_u\subseteq Even (C_u)$. 
\qed \end{proof}


\begin{theorem}
For any even $n>n_0$,  there exists a non-bipartite $\lfloor 0.110n \rfloor$-odd dominated graph of order $n$. 
\end{theorem}


\begin{proof}
Given $n$,  $r\le n$ s.t. $r=n \mod 2$, and $k\ge 0$.
Let $p=(n-r)/2$, and  
let $G=(V_0\cup V_1\cup V_2,E)$ s.t. $|V_0|=|V_1| = p$, $|V_2|=r$ be a random graph on $n$ vertices s.t. for any $u\in V_i$, $v\in V_j$ there is an edge between $u$ and $v$ with probability $0$ if $i=j$ and with probability $1/2$ otherwise.

For any $i\in \{0,1\}$, and any non empty $D\subseteq V\setminus V_i$ s.t. $|D|\le k$, 
let $A^{(i)}_D$ be the bad event $|Odd_G(D)\cap V_{i}|\le k-|D|$. Since each vertex of $V_i$ is in $Odd_G(D)$ with probability $1/2$, 
$Pr(A^{(i)}_D)= \sum_{j=0}^{k-|D|}{p\choose j}2^{-p}\le 2^{p[H(\frac {k-|D|}{p})-1]}$.

Another bad event is that $G$ is bipartite which occurs with probability less than $(\frac 7 8)^{pr}$. Indeed,  the probability that given $u\in V_0,v\in V_1, w\in V_2$, $(u,v,w)$ do not form a triangle is $\frac 78$, so given a bijection $f: V_0\to V_1$, the probability that  $\forall u\in V_0, \forall w\in V_2$,   $(u,f(u),w)$ do not form a triangle is $(\frac 7 8)^{pr}$.




Let $X$ be the number of bad events. 
\begin{eqnarray*}
E[X] = 2\sum_{d=1}^{k}{p+r\choose d}  \sum_{j=0}^{k-d}{p\choose j}2^{-p} +(\frac 7 8)^{pr} \\
\le 2\sum_{d=1}^{k}2^{(p+r)H(\frac{d}{p+r})+pH(\frac{k-d}{p})-p}+(\frac 7 8)^{pr} \\
\le 2\sum_{d=1}^{k}2^{pH(\frac{d}{p+r})+pH(\frac{k-d}{p})-p + r}+(\frac 7 8)^{pr} \\
\end{eqnarray*}
The function $d\mapsto pH(\frac d{p+r})+pH(\frac{k-d}p)-p+r$ is maximal for $d=\frac {k(p+r)} {2p+r}$. Thus,
  
$E[X] 
\le 2k2^{2pH(\frac{k}{2p+r})-p + r}+(\frac 7 8)^{pr} 
$.

By taking $r=1$, and $k=0.11n=0.11(2p+1)$, $E[X]<1$ when $p$ large enough, thus $G$  has no bad event with a non zero probability. 
\qed \end{proof}

\begin{corollary}
There exist  contextuality scenarios with linear multipartiteness width: for any even $n>n_0$,  there exist graph games on $n$ players producing  contextuality scenarios  of multipartiteness width at least $\lfloor 0.11n \rfloor$.
\end{corollary}
\begin{proof}
Using the result from \cite{AM}, for any non bipartite graph of order $n$ being $0.11n$-o.d ensures  that the probability distribution obtained by using the quantum strategy cannot be simulated using non local boxes  involving at most 0.11$n$ parties. Thus  lemma \ref{lem:simwin} allows to conclude that the associated pseudo-telepathy game cannot be won classically. Therefore there is no interpretation that is $k$-multipartite with $k<0.11n$  which means that the  contextuality scenario has linear width.
\qed \end{proof}

\section{Conclusion}
We have shown that there exist graphs with linear multipartiteness width, however the proof is non constructive and  the  best known bound for explicit families  is logarithmic.
 A natural future direction of research would be to 
find explicit families with linear multipartiteness width or to improve the bounds proven for the Paley graph states.  An other important question is to consider lower bounds for the  scenarios associated with the graph games.
 A promising area of investigation for multipartite scenarios is: what happens if we limit the width of shared randomness?  Indeed,  for the proof of how winning  the game allows to simulate the quantum probability distributions, one needs only shared random variables that are correlated in local neighborhoods in the graph.
 One can also consider the link with building entanglement witnesses for graph states, generalizing the construction of \cite{HD06}.
It would be also very interesting to link the  multipartiteness width with the structures of the groups of  the associated binary linear system defining the two-player bipartite non-local games \cite{S}.
 Finally, 
one can expect that the multipartiteness width of the Paley graph states might have cryptographic applications to ensure security against cheating for some protocols for example.


%

\begin{thebibliography}{9}
 
 
\bibitem{AB} S. Abramsky  and A. Brandenburger
\newblock {\em The sheaf theoretic structure of non locality and contextuality.}
\newblock New Journal of Physics 13  113036 (2011)

\bibitem{shane} S. Abramsky, R.S.  Barbosa and S. Mansfield
  \newblock {\em Quantifying contextuality via linear programming},
  \newblock Informal Proceedings of Quantum Physics \& Logic, (2016)

\bibitem{AFLS} A. Ac\'in, T. Fritz, A. Leverrier,  and A. Bel\'en Sainz
\newblock {\em A combinatorial approach to nonlocality and contextuality.}
\newblock  Comm. Math. Phys. 334(2), 533-628 (2015) 

\bibitem{AM} A. Anshu and M. Mhalla
\newblock {\em Pseudo-telepathy games and genuine NS k-way nonlocality using graph states}
\newblock Quantum Information and Computation, Vol. 13, No. 9, 10  0833-0845 Rinton Press (2013)



\bibitem{COLT}A. Badanidiyuru, J. Langford, A. Slivkins
\newblock {\em  Resourceful Contextual Bandits.} 
\newblock http://arxiv.org/abs/1402.6779, COLT (2014)

\bibitem{BP}J. Barrett, S. Pironio, J. Bancal and N. Gisin, 
\newblock{ \em The definition of multipartite nonlocality.}
\newblock  Phys. Rev. A 88, 014102 (2013)

\bibitem{Bell}J.S. Bell
\newblock{ \em On the Einstein-Podolsky-Rosen paradox.}
\newblock Physics 1 195-200 (1964)

\bibitem{BP} J. Barrett and S. Pironio, 
\newblock{ \em  Popescu-Rohrlich correlations as a unit of nonlocality.}
\newblock  Phys. Rev. Lett. 95, 140401, (2005)

\bibitem{BB}G. Brassard, A. Broadbent and  A. Tapp
\newblock {\em Multi-Party Pseudo-Telepathy}
\newblock  8th International Workshop, WADS 2003, Ottawa, Ontario, Canada, July 30 - August 1,  Proceedings (2003).

\bibitem{BM} A. Broadbent and A. A. Methot, 
\newblock {\em On the power of non-local boxes}
\newblock  Theoretical Computer Science C 358: 3-14, (2006)


\bibitem{Bob}B. Coecke
\newblock{ \em From quantum foundations via natural language meaning to a theory of everything.}
\newblock   arXiv:1602.07618v1  (2016)




 \bibitem{livre}E. Dzhafarov , S. Jordan , R. Zhang , V. Cervante
\newblock {\em Contextuality from Quantum Physics to Psychology.}
\newblock  Advanced Series on Mathematical Psychology: Volume 6 (2016)

 \bibitem{FR}D. J. Foulis and C. H. Randall, 
 \newblock {\em Empirical logic and tensor products}.
 \newblock  J. Mathematical Phys. 5 , 9:20. MR683888. (1981)
 
\bibitem{AL}T. Fritz, A. B. Sainz, R. Augusiak, J. Bohr Brask, R. Chaves, A. Leverrier, and A. Ac\'in
\newblock {\em Local orthogonality as a multipartite principle for quantum correlations}
\newblock Nature Communications 4, 2263 (2013)


\bibitem{qc}A. Grudka, K. Horodecki, M. Horodecki, P. Horodecki, R. Horodecki, P. Joshi, W. Klobus, A. W\'ojcik
\newblock {\em Quantifying contextuality}
\newblock Phys. Rev. Lett. 112, 120401 (2014)

\bibitem{GHZ}D. M. Greenberger, M. A. Horne, A. Shimony, and A. Zeilinger, 
\newblock {\em Bell's theorem without inequalities.} 
\newblock American Journal of Physics 58 , 1131.(1990)


\bibitem{GRRH}P. Gnaci\'nski, M. Rosicka, R. Ramanathan, K. Horodecki, M. Horodecki, P. Horodecki, S. Severini
\newblock {\em Linear game non-contextuality and Bell inequalities - a graph-theoretic approach.}
\newblock e-print arXiv:1511.05415 Nov (2015)


\bibitem{AGT}C. Godsil, G. Royle. 
\newblock {\em Algebraic graph theory, volume 207 of Graduate Texts in Mathematics.}
\newblock  volume 207 of Graduate Texts in Mathematics - Springer-Verlag, (2001) 


\bibitem{GTB05}O. G\"uhne, G.T\'oth, H.J. Briegel
\newblock {\em Multipartite entanglement in spin chains.}
\newblock New Journal of Physics, 7(1), 229.  (2005)


\bibitem{Hardy}L. Hardy
\newblock {\em  Nonlocality for two particles without inequalities for almost all entangled states.}
\newblock  Physical Review Letters 71 , 1665:1668.(1993)

 
\bibitem{HD06} M. Hein,W. D\"ur, J. Eisert, R. Raussendorf, M. Nest and  H. J.  Briegel
\newblock {\em  Entanglement in graph states and its applications.}
\newblock  arXiv preprint quant-ph/0602096.  (2006)

\bibitem{Context}M. Howard, J.J. Wallman, V. Veitch and  J. Emerson
\newblock {\em Contextuality supplies the "magic" for quantum computation.}
\newblock    Nature 510, 351-355 (2014)


\bibitem{K}V. Kumar. 
\newblock {\em Algorithms for constraint-satisfaction problems: A survey.}
\newblock  AI magazine, 13(1):32, (1992).


\bibitem{LR}R. Luce and H. Rai.
\newblock {\em  Games and Decisions: Introduction and Critical Survey.}
\newblock  Doverbooks on advanced mathematics, Dover Publications (1957)

\bibitem{LRB14} Y.-C.Liang, D. Rosset, J-D. Bancal, G.P\"utz, T.J. Barnea and N. Gisin
\newblock {\em Family of Bell-like inequalities as device-independent witnesses for entanglement depth.}
\newblock Physical review letters, 114(19), 190401. (2015)



\bibitem{Mermin}N. D. Mermin, 
\newblock {\em Extreme quantum entanglement in a superposition of macroscopically distinct states.}
\newblock Physical Review Letters, 65(15):1838-1849, (1990).

\bibitem{MP}M. Mhalla and S. Perdrix 
\newblock {\em Graph States, Pivot Minor, and universality of (X-Z) Measurements.}
\newblock  IJUC 9 (1-2) : 153-171, (2013).

\bibitem{S} W. Slofstra
\newblock {\em Tsirelson's problem and an embedding theorem for groups arising from non-local games}
\newblock e-print	arXiv:1606.03140 (2106)

\bibitem{NFP}N. Pflueger.
\newblock {\em Graph reductions, binary rank, and pivots in gene assembly.}
\newblock Discrete Applied Mathematics, 2011.



\bibitem{cwc}W. Zeng and P. Zahg,
\newblock {\em Contextuality and the weak axiom of the theory of choice}
\newblock  Quantum Interactions proc. Volume 9535 of the series Lecture Notes in Computer Science pp 24-35, (2016)





\end{thebibliography}

\section*{Acknowledgements.}
We would like to thanks an anonymous reviewer for noticing a mistake in an earlier version and helpful comments.

\end{document}